\documentclass{llncs}

\usepackage{amsmath}
\usepackage{amssymb}
\usepackage{graphicx}
\usepackage{tikz}
\usetikzlibrary{arrows}

\title{Sums of read-once formulas: How many summands suffice?}

\author{
Meena Mahajan 
\and {Anuj Tawari} 
}

\institute{%
  The Institute of Mathematical Sciences, Chennai, India.
  \email{\{meena,anujvt\}@imsc.res.in}
}

\begin{document}
\maketitle
\begin{abstract}
An arithmetic read-once formula (ROF) is a formula (circuit of fan-out
1) over 
$+, \times$ where each variable labels at most one leaf.
Every multilinear polynomial can be expressed as the sum of ROFs. 
In this work, we prove, for certain multilinear polynomials, 
a tight lower bound on the number of summands in such an expression.
\end{abstract}



\newtheorem{observation}[theorem]{Observation}
\newtheorem{claimn}[theorem]{Claim}
\newtheorem{fact}[theorem]{Fact}

\newcommand{\E}{{\mathbb E}}
\newcommand{\C}{{\mathbb C}}
\newcommand{\D}{{\mathcal D}}
\newcommand{\M}{{\mathcal M}}
\newcommand{\bigoh}{{\mathrm{O}}}
\newcommand{\smalloh}{{\mathrm{o}}}

\newenvironment{proof-sketch}{\noindent{\bf Sketch of Proof}\hspace*{1em}}{\qed\bigskip}


\newcommand{\enc}{{\sf Enc}}
\newcommand{\dec}{{\sf Dec}}
\newcommand{\Var}{{\rm Var}}
\newcommand{\Deg}{{\rm deg}}
\newcommand{\ROF}{{\rm ROF}}
\newcommand{\ROP}{{\rm ROP}}
\newcommand{\Z}{{\mathbb Z}}
\newcommand{\X}{{\mathbf x}}
\newcommand{\F}{{\mathbb F}}
\newcommand{\N}{{\mathbb N}}
\newcommand{\I}{{\mathcal I}}
\newcommand{\integers}{{\mathbb Z}^{\geq 0}}
\newcommand{\R}{{\mathbb R}}
\newcommand{\Q}{{\cal Q}}
\newcommand{\V}{{\mathbb V}}
\newcommand{\eqdef}{{\stackrel{\rm def}{=}}}
\newcommand{\from}{{\leftarrow}}
\newcommand{\vol}{{\rm Vol}}
\newcommand{\poly}{{\rm poly}}
\newcommand{\ip}[1]{{\langle #1 \rangle}}
\newcommand{\wt}{{\rm wt}}
\renewcommand{\vec}[1]{{\mathbf #1}}
\newcommand{\mspan}{{\rm span}}
\newcommand{\rs}{{\rm RS}}
\newcommand{\RM}{{\rm RM}}
\newcommand{\Had}{{\rm Had}}
\newcommand{\calc}{{\cal C}}
\renewcommand{\epsilon}{\varepsilon}
\renewcommand{\phi}{\varphi}

\section{Introduction}
\label{sec:intro}

Read-once formulas (ROF) are formulas (circuits of fan-out 1) in which
each variable appears at most once. A formula computing a polynomial
that depends on all its variables must read each variable at least
once. Therefore, ROFs compute some of the simplest possible functions that
depend on all of their variables. The polynomials computed by such
formulas are known as read-once polynomials (ROPs). Since every
variable is read at most once, ROPs are multilinear
 \footnote{A polynomial is said to be multilinear if the individual degree of 
 each variable is at most one.}. But not every
multilinear polynomial is a ROP. For example, $x_{1}x_{2} + x_{2}x_{3}
+ x_{1}x_{3}$.  

We investigate the following question: Given an $n$-variate
multilinear polynomial, can it be expressed as a sum of at most $k$
ROPs? It is easy to see that every bivariate multilinear polynomial is
a ROP. Any tri-variate multilinear polynomial can be expressed as a
sum of 2 ROPs. With a little thought, we can obtain a sum-of-3-ROPs
expression for any 4-variate multilinear polynomial.  An easy
induction on $n$ then shows that any $n$-variate multilinear
polynomial, for $n \ge 4$, can be written as a sum of at most $3
\times 2^{n-4}$ ROPs. Also, the sum of two multilinear monomials is a
ROP, so any $n$-variate multilinear polynomial with $M$ monomials can
be written as the sum of $\lceil M/2 \rceil$ ROPs.  
We ask the following question: Does there exist a
strict hierarchy among $k$-sums of ROPs?  Formally,
\begin{problem}
  Consider the family of $n$-variate multilinear polynomials. For
  $1 < k \le 3 \times 2^{n-4}$, is $\sum^{k} \cdot \ROP$ strictly more powerful than   
  $\sum^{k-1} \cdot \ROP$? If so, what explicit polynomials witness
  the separations?
\end{problem}
We answer this affirmatively for $k \le \lceil n/2 \rceil$. In
particular, for $k = \lceil n/2 \rceil$, there exists an explicit
$n$-variate multilinear polynomial which cannot be written as a sum of
less than $k$ ROPs but it admits a sum-of-$k$-ROPs representation.

Note that $n$-variate ROPs are computed by linear sized formulas. Thus
if an $n$-variate polynomial $p$ is in $\sum^k \cdot \ROP$, then $p$
is computed by a formula of size $O(kn)$ where every intermediate node
computes a multilinear polynomial. Since superpolynomial lower bounds
are already known for the model of multilinear formulas \cite{Raz09},
we know that for those polynomials (including the determinant and the
permanent), a $\sum^k \cdot \ROP$ expression must have $k$ at least
quasi-polynomial in $n$. However the best upper bound on $k$ for these
polynomials is only
exponential in $n$, leaving a big gap between the lower and upper
bound.   On the other hand, our lower bound is provably tight.

A counting argument shows that a random multilinear polynomial
requires exponentially many ROPs; there are multilinear polynomials
requiring $k = \Omega(2^n/n^2)$.  Our general upper bound on $k$ is
$O(2^n)$, leaving a gap between the lower and upper bound. One
challenge is to close this gap. A perhaps more interesting challenge
is to find explicit polynomials that require exponentially large $k$
in any $\sum^{k}\cdot \ROP$ expression. 

A natural question to ask is whether stronger lower bounds than the
above result can be proven. 
In particular, to separate 
$\sum^{k-1}\cdot \ROP$  from $\sum^{k}\cdot \ROP$, how many variables
are needed? The above hierarchy result says that 
$2k-1$ variables suffice, but there may be simpler polynomials (with
fewer variables) witnessing this separation. 
We demonstrate another technique which
improves upon the previous result for $k=3$, showing that 4 variables suffice. In particular, we show
that over the field of reals, there exists an explicit multilinear
$4$-variate multilinear polynomial which cannot be written as a sum
of $2$ ROPs. This lower bound is again tight, as there is a sum of
$3$ ROPs representation for 
every 4-variate multilinear polynomial. 

\subsection*{Our results and techniques}
\label{sec:our-results}

We now formally state our results.

\begin{theorem}
  \label{thm:hierarchy}
For each $n \ge 1$, the $n$-variate degree $n-1$ symmetric polynomial
$S_n^{n-1}$ cannot be written as a sum of less than $\lceil n/2
\rceil$ ROPs, but it can be written as a sum of $\lceil n/2 \rceil$
ROPs.
\end{theorem} 

The idea behind the lower bound is that if $g$ can be expressed as a sum of
less than $\lceil n/2 \rceil$ ROFs, then one of the ROFs can be
eliminated by taking partial derivative with respect to one variable
and substituting another by a field constant. We then use the
inductive hypothesis to arrive at a contradiction.  This approach
necessitates a stronger hypothesis than the statement of the theorem,
and we prove this stronger statement in Lemma~\ref{lem:lower-bound}
as part of Theorem~\ref{thm:hierarchy-general}.

\begin{theorem}
  \label{thm:hard-for-2-sum}
There is an explicit $4$-variate multilinear polynomial $f$ which
cannot be written as the sum of $2$ ROPs over $\R$. 
\end{theorem}

The proof of this theorem mainly relies on a structural lemma
(Lemma~\ref{lem:structure-2-sum-4-variate}) for sum of $2$ read-once
formulas. In particular, we show that if $f$ can be written as a sum
of $2$ ROPs then one of the following must be true: 1.~Some 2-variate
restriction is a linear polynomial.  2.~There exist variables
$x_{i}, x_{j} \in \Var(f)$ such that the polynomials
$x_{i}, x_{j}, \partial_{x_{i}}(f), \partial_{x_{j}}(f), 1$ are linearly
dependent. 3.~We can represent $f$ as
$f = l_{1} \cdot l_{2} + l_{3} \cdot l_{4}$ where
$(l_{1}, l_{2})$ and $(l_{3}, l_{4})$ are
variable-disjoint linear forms. Checking the first two conditions is
easy. For the third condition we use the commutator of $f$, introduced 
in \cite{SV14}, to find one of the $l_{i}$'s. The knowledge of one
of the $l_{i}$'s suffices to determine all the linear forms. Finally,
we construct a $4$-variate polynomial which does not satisfy any of
the above mentioned conditions. This construction does not work over
algebraically closed fields. We do not yet know how to construct an
explicit 4-variate multilinear polynomial not expressible as the sum
of 2 ROPs over such fields, or even whether such polynomials exist.  

\subsection*{Related work}
\label{sec:related-work}

Despite their simplicity, ROFs have received a lot of attention both in the
arithmetic as well as in the Boolean world \cite{HH91,BHH95a,BB98,BC98,SV14,SV15}. The most
fundamental question that can be asked about polynomials is the
polynomial identity testing (PIT) problem: Given an arithmetic circuit
$\mathcal{C}$, is the polynomial computed by $\mathcal{C}$ identically
zero or not. PIT has a randomized polynomial time algorithm: Evaluate
the polynomial at random points. It is not known whether PIT has a
deterministic polynomial time algorithm. In 2004, Kabanets and
Impagliazzo established a connection between PIT algorithms and
proving general circuit lower bounds \cite{KI04}. However, for
restricted arithmetic circuits, no such result is known. For instance,
consider the case of multilinear formulas. Even though strong lower
bounds are known for this model, there is no efficient deterministic
PIT algorithm. For this reason, PIT was studied for the weaker model
of sum of read-once formulas. Notice that multilinear depth $3$
circuits are a special case of this model. 

Shpilka and Volkovich gave a deterministic PIT algorithm for the sum
of a small number of ROPs \cite{SV15}. Interestingly, their proof uses
a lower bound for a weaker model, that of $0$-justified ROFs  (setting
some variables to zero does not kill any other variables). In
particular, they show that the polynomial $\mathcal{M}_{n} =
x_{1}x_{2} \cdots x_{n}$, consisting of just a single monomial, cannot
be represented as a sum of less than $n/3$ weakly justified ROPs. More
recently, Kayal showed that if $\mathcal{M}_{n}$ is represented as a
sum of powers of low degree (at most $d$) polynomials, then the number
of summands is at most $\exp(\Omega(n/d))$ \cite{Kayal15}. He used this
lower bound to give a PIT algorithm. Our lower bound from
Theorem~\ref{thm:hierarchy} is orthogonal to both these results and is
provably tight. An interesting question is whether it can be used
to give a PIT algorithm.

Similar to ROPs, one may also study read-restricted formulas. 
For any number $k$, R$k$Fs are formulas that read every variable at
most $k$ times. For $k > 1$, R$k$Fs for $k\ge 2$ need not be
multilinear, and thus are strictly more powerful than ROPs. However,
even when restricted to multilinear polynomials, they are more
powerful; in \cite{AMV15}, Anderson, Melkebeek and Volkovich show that
there is a multilinear $n$-variate polynomial in R2F requiring $\Omega(n)$
summands when written as a sum of ROPs.

\subsection{Organization}
\label{sec:org}
The paper is organized as follows. In Section~\ref{sec:prelim} we give
the basic definitions and notations. In Section~\ref{sec:hierarchy},
we establish Theorem~\ref{thm:hierarchy}. 
showing that the hierarchy of $k$-sums of ROPs is proper. 
In Section~\ref{sec:hard-4-variate} we establish
Theorem~\ref{thm:hard-for-2-sum}, showing an explicit 4-variate
multilinear polynomial that is not expressible as the sum of two
ROPs. We conclude in 
Section~\ref{sec:concl} with some further questions that are still open.

\section{Preliminaries}
\label{sec:prelim}

For a positive integer $n$, we denote $ [n] = \{ 1,2, \ldots,
n\}$. For a polynomial $f$, by $\Var(f)$ we mean the set of variables
occurring in $f$. For a polynomial $f(x_{1}, x_{2}, \ldots, x_{n})$, a
variable $x_{i}$ and a field element $\alpha$, we denote by
$f \mid_{x_{i} = \alpha}$ the polynomial resulting from setting
$x_{i}= \alpha$. Let $f$ be an $n$-variate polynomial. We say that $g$
is a $k$-variate restriction of $f$ if $g$ is obtained by setting some
variables in $f$ to field constants and $|\Var(g)| \leq k$. A set of
polynomials $f_{1}, f_{2}, \ldots, f_{k}$ over the field $\F$ is said
to be linearly dependent if there exist constants $\alpha_{1},
\alpha_{2}, \ldots, \alpha_{k}$ such that $\displaystyle\sum_{i \in
  [k]} \alpha_{i} f_{i}  =0$. 

The $n$-variate degree $k$ elementary symmetric polynomial, denoted
$S_n^k$, is defined as follows:
\[ S_n^k(x_1, \ldots , x_n) = \sum_{A\subseteq [n], |A|=k} ~~~~\prod_{i\in
A} x_i.\]

A circuit is a directed acyclic graph with variables and field
constants labeling the leaves, field operations $+, \times$ labeling
internal nodes, and a designated sink node. Each node naturally
computes a polynomial; the polynomial at the designated sink node is
the polynomial computed by the circuit. If the underlying undirected
graph is a tree, then the circuit is called a formula. A formula
is said to be read-$k$ if each variable appears as  a leaf label at
most  $k$ times.

For read-once formulas, it is more convenient to use the following
``normal form'' from \cite{SV15}. 

\begin{definition}[Read-once formulas \cite{SV15}]
  \label{def:ROF}
A read-once arithmetic formula (ROF) over a field $\F$ in the
variables  $\{x_{1}, x_{2}, \allowbreak \ldots, x_{n}\}$ is a binary
tree as follows. The leaves are labeled by variables and internal
nodes by $\{+, \times\}$. In addition, every node is labeled by a pair
of field elements $(\alpha, \beta) \in \F^{2}$. Each input variable
labels at most once leaf. The computation is performed in the
following way. A leaf labeled by $x_{i}$ and $(\alpha, \beta)$
computes $\alpha x_{i} + \beta$. If a node $v$ is labeled by
$\star \in \{+, \times\}$ and $(\alpha, \beta)$ and its children compute the
polynomials $f_{1}$ and $f_{2}$, then $v$ computes
$\alpha (f_{1} \star f_{2}) + \beta$. 
\end{definition}
We say that $f$ is a read-once polynomial (ROP) if it can be computed
by a ROF, and is in $\sum^k \cdot \ROP$ if it can be
expressed as the sum of at most $k$ ROPs. 

\begin{proposition}
\label{prop:easy-upper-bound}
For every $n$, every $n$-variate multilinear polynomial can be written
as the sum of at most $\lceil 3 \times 2^{n-4}\rceil$ ROPs.
\end{proposition}
\begin{proof}
For $n=1,2,3$  this is easy to see. 

For $n=4$, let $f(X)$ be given by the expression $\sum_{S \subseteq
[4]} a_S x_S$, where $x_S$ denotes the monomial $\prod_{i \in S}x_i$. 
We want to express $f$ as 
$f_1+f_2+f_3$, where each $f_i$ is an ROP. 
If there are no degree $2$ terms, we use the following: 
\begin{eqnarray*}
f_1 &=& a_\emptyset + a_1x_1 + a_2x_2 + a_3x_3 +a_4x_4 \\
f_2 &=& x_1x_2(a_{123}x_3 +a_{124}x_4) \\
f_3 &=& x_3x_4(a_{134}x_1 +a_{234}x_2 + a_{1234}x_1x_2)
\end{eqnarray*}
Otherwise, assume without loss of generality that $a_{13} \neq 0$. Then define
\begin{eqnarray*}
f_1 &=& \left[\displaystyle\sum_{S \subseteq [2]}
 a_{S} \displaystyle\prod_{i \in S} x_i\right]
 + \left[ \displaystyle\sum_{\emptyset\neq S \subseteq \{3,4\}}
 a_{S} \displaystyle\prod_{i \in S} x_i \right] \\
f_2 &=&  \left(a_{13} x_1 + a_{23} x_2 + a_{123} x_1x_2\right) \cdot
 \left(\frac{a_{14}}{a_{13}} x_4 + x_3 + \frac{a_{134}}{a_{13}}
 x_3x_4 \right) \\
f_3 &=&  x_2x_4 \left[
\left(a_{24} - \frac{a_{14}a_{23}}{a_{13}}\right) 
+ x_1 \left(a_{124} - \frac{a_{14}a_{123}}{a_{13}}\right) 
\right. \\
&& \left.
~~~~~~~~+ x_3\left(a_{234} - \frac{a_{134}a_{23}}{a_{13}}\right) 
+ x_1x_3\left(a_{1234} - \frac{a_{134}a_{123}}{a_{13}}\right)
\right]
 \end{eqnarray*}
Since any bivariate multilinear polynomial is a ROP, each $f_i$ is
indeed an ROP.

For $n > 4$, express $f$ as $x_n g + h$ where $g = \partial_{x_n}f$
and $h=f\mid _{x_n=0}$, and use induction, along with the fact that $g$
does not have variable $x_n$.
\qed\end{proof}

\begin{proposition}
\label{prop:easy-upper-bound2}
For every $n$, every $n$-variate multilinear polynomial with $M$
monomials can be written
as the sum of at most $\lceil \frac{M}{2} \rceil$ ROPs.
\end{proposition}
\begin{proof}
For $S\subseteq [n]$, let $x_S$ denote the multilinear monomial
$\prod_{i\in S}x_i$. For any $S, T \subseteq [n]$, the polynomial $ax_S
+ bx_T$ equals $x_{S\cap T}(ax_{S \setminus T} + bx_{T\setminus S})$ and
hence is an ROP. Pairing up monomials in
any way  gives the $\lceil \frac{M}{2} \rceil$ bound. 
\end{proof}

The partial derivative of a polynomial is defined naturally over
continuous domains. The definition can be extended in more than one
way over finite fields. However, for multilinear polynomials, these
definitions coincide. We consider only multilinear polynomials in this
paper, and the following formulation is most useful for us: The
partial derivative of a polynomial $p\in \F[x_{1}, x_{2}, \ldots ,
  x_{n}]$ with respect to a variable $x_i$, for $i\in [n]$, is given
by $\partial_{x_{i}}(p) \triangleq p \mid_{x_i=1} - p \mid_{x_i=0}$.
For multilinear polynomials, the sum, product, and chain rules
continue to hold. 

\begin{fact}[Useful Fact about ROPs \cite{SV15}]
\label{fact:ROP}
The partial derivatives of ROPs are also ROPs.
\end{fact}

\begin{proposition}[3-variate ROPs]
\label{prop:3-var-ROP}
Let $f \in \F[x_{1}, x_{2}, x_{3}]$ be a $3$-variate ROP. Then there
exists $i \in [3]$ and $a \in \F$ such that $\deg(f\mid_{x_{i}= a})
\leq 1$. 
\end{proposition}

\begin{proof}
Assume without loss of generality that $f = f_{1}(x_{1}) \star
f_{2}(x_{2}, x_{3}) + c$ where $\star \in \{+, \times\}$ and $c \in
\F$. If $\star = +$, then for all $a \in \F$, $\deg(f\mid_{x_{2}=a})
\leq 1$. If $\star = \times$, $\deg(f\mid_{f_{1} = 0})\leq 1$.  
\qed \end{proof}

We will also be dealing with a special case of ROFs called
multiplicative ROFs defined below: 

\begin{definition}[Multiplicative Read-once formulas]
\label{def:m-rof}
  A ROF is said to be a multiplicative ROF if it does not contain any
addition gates. We say that $f$ is a multiplicative ROP if it can be
computed by a multiplicative ROF. 
\end{definition}

\begin{fact}[\cite{SV15} (Lemma 3.10)]
\label{fact:mrop-partial}
A ROP $p$ is a  multiplicative ROP if and only if for any two
variables $x_i, x_j \in \Var(p)$, $\partial_{x_i}\partial_{x_j}(p) \neq 0$. 
\end{fact}

Multiplicative ROPs have the following useful property, observed
in \cite{SV15}. (See Lemma 3.13 in \cite{SV15}. For completeness, and
since we refer to the proof later, we include a proof sketch here.)  
\begin{lemma}[\cite{SV15}]
  \label{lem:mrops}
Let $g$ be a multiplicative ROP with $|\Var(g)| \geq 2$. For every
$x_{i} \in \Var(g)$, there exists $x_{j} \in \Var(g)\setminus \{x_i\}$
and $\gamma \in \F$ 
such that $\partial_{x_{j}}(g) \mid_{x_{i} = \gamma} = 0$.
\end{lemma}
\begin{proof}
Let $\phi$ be a multiplicative ROF computing $g$. 
Pick any $x_i \in \Var(g)$.
As $|\Var(\phi)| =
|\Var(g)| \geq 2$, $\phi$ has at least one gate. Let $v$ be the unique
neighbour (parent) of the leaf labeled by $x_{i}$, and let $w$ be the
other child of $v$. We denote by
$P_{v}(\bar{x})$ and $P_{w}(\bar{x})$ the ROPs  computed by $v$ and
$w$. Since $v$ is a $\times$ gate and we  use the normal form 
from Defintion~\ref{def:ROF}, $P_v$  is of the form 
$(\alpha x_i + \beta) \times P_w$ for some $\alpha \neq 0$. 

Replacing the output from $v$ by a new variable $y$, we obtain
from $\phi$ another multiplicative ROF $\psi$ in the variables
$\{y\} \cup \Var(g) \setminus \Var(P_v)$. Let $\psi$ compute the
polynomial $Q$; then $g = Q \mid_{y=P_v}$. 

Note that the sets $\Var(Q), \{x_i\}, \Var(P_w)$ are non-empty 
and disjoint, and form   a partition of $\{y, x_1, \ldots , x_n \}$.

By the chain rule, for every variable $x_{j} \in \Var(P_w)$ we have:
$$ \partial_{x_{j}}(g) = \partial_{y}(Q) \cdot
\partial_{x_{j}}(P_{v}) = \partial_{y}(Q) \cdot 
(\alpha x_i + \beta) \cdot
\partial_{x_{j}}(P_w) $$ 
It follows that for $\gamma = -\beta/\alpha $,
$\partial_{x_{j}}(g)\mid_{x_{i} = \gamma} = 0$.   
\qed \end{proof}

Along with partial derivatives, another operator that we will find
useful is the commutator of a polynomial.  The commutator of a
polynomial has previously been used for polynomial factorization and
in reconstruction algorithms for read-once formulas, see 
 \cite{SV14}.

\begin{definition}[Commutator \cite{SV14}]
  \label{def:commutator}
Let $P \in \F[x_{1}, x_{2}, \ldots, x_{n}]$ be a multilinear
polynomial and let $i, j \in [n]$. The commutator between
$x_{i}$ and $x_{j}$, denoted $\triangle_{ij}P$, is defined as follows.
$$ \triangle_{ij}P =
\left( P\mid_{x_{i}=0, x_{j}=0} \right) \cdot
\left( P\mid_{x_{i}=1, x_{j}=1} \right)
- \left( P\mid_{x_{i}=0, x_{j}=1} \right) \cdot
\left( P\mid_{x_{i}=1, x_{j}=0} \right)$$  
\end{definition}

The following property of the commutator will be useful to us.

\begin{lemma}
  \label{lem:commutator}
Let $f = l_{1}(x_{1}, x_{2}) \cdot l_{2}(x_{3}, x_{4}) + l_{3}(x_{1},
x_{3}) \cdot l_{4}(x_{2}, x_{4})$ where the $l_{i}$'s are linear
polynomials. Then $l_{2}$ divides $\triangle_{12}(f)$. 
\end{lemma}
\begin{proof}
First, we show that $\triangle_{12}(l_{3} \cdot l_{4}) = 0$. Assume
$l_{3} = C x_{1} + m$ and $l_{4} = D x_{2} + n$ where $C, D \in \F$
and $m, n$ are linear polynomials in $x_3$, $x_4$ respectively. 
By definition,
$\triangle_{12}(l_{3} \cdot l_{4}) = mn(C + m)(D + n) -
m(D + n)(C + m)n = 0$. 

Now we write $\triangle_{12}f$ explicitly.
Let $l_{1} = ax_{1} + bx_{2} + c$. By definition,  
\begin{align*}
\begin{aligned}
\triangle_{12}f &= \triangle_{12}(l_{1}l_{2} + l_{3}l_{4}) \\
&= (cl_{2}+ mn)((a+b+c)l_{2} + (C + m)(D + n)) - \\
&~~~ ((b+c)l_{2}+m(D + n))\cdot ((a+c)l_{2}+ n(C + m)) \\ 
&= l_{2}^{2} (c(a+b+c) - (a+c)(b+c))  \\
&~~~ + l_{2}(c(C + m)(D + n) +  mn(a+b+c) - n(b+c)(C + m) -m(a+c)(D + n)) 
 \end{aligned}
\end{align*}

It follows that $l_{2}$ divides $\triangle_{12}f$.
\qed \end{proof}

\section{A proper hierarchy in $\sum^k \cdot \textrm{ROP}$ }
\label{sec:hierarchy}
This section is devoted to proving Theorem~\ref{thm:hierarchy}.


We prove the lower bound for $S_n^{n-1}$ by induction. This
necessitates a stronger induction hypothesis,  so we will actually
prove the lower bound for a larger class of polynomials. 
For any $\alpha,\beta \in \F$, we define the polynomial
$\M_n^{\alpha,\beta} = \alpha S_{n}^{n} + \beta S_{n}^{n-1}$. 
We note the following recursive structure of $\M_n^{\alpha,\beta}$:
\begin{eqnarray*}
(\M_{n}^{\alpha,\beta}) \mid_{x_n=\gamma}  &=&
\M_{n-1}^{\alpha\gamma+\beta,\beta\gamma} ~~.  \\
\partial_{x_{n}}(\M_{n}^{\alpha,\beta})  &=& \M_{n-1}^{\alpha,\beta} ~~. 
\end{eqnarray*}

We show below that each $\M_n^{\alpha,\beta}$ is expressible as the
sum of $\lceil n/2 \rceil$ ROPs (Lemma~\ref{lem:upper-bound});
however, for any non-zero $\beta \neq 0$, $\M_n^{\alpha,\beta}$ cannot
be written as the sum of fewer than $\lceil n/2 \rceil$ ROPs
(Lemma~\ref{lem:lower-bound}).   At $\alpha=0$, $\beta=1$, we get
$S^{n-1}_n$, the simplest such polynomials, 
establishing Theorem~\ref{thm:hierarchy}.

\begin{lemma}
  \label{lem:lower-bound}
  Let $\F$ be a field. For every $\alpha \in \F$ and $\beta \in \F
  \setminus \{0\}$, the polynomial $\M_{n}^{\alpha,\beta} =
  \alpha S_{n}^{n} + \beta S_{n}^{n-1}$ cannot be written as a sum of
  $k < n/2$ ROPs.
\end{lemma}
\begin{proof}
The proof is by induction on $n$. The cases $n = 1,2$ are easy to
see. We now assume that $k \geq 1$ and $n > 2k$. Assume to the
contrary that there are ROPs $f_{1}, f_{2}, \ldots, f_{k}$ over
$\F[x_{1}, x_{2}, \ldots, x_{n}]$ such that
$\displaystyle f \triangleq \sum_{m \in  [k]} f_{m}
= \M_{n}^{\alpha,\beta}$.
The main steps in the  proof are as follows:
\begin{enumerate}
\item Show using the inductive hypothesis that for all $m \in [k]$ and
  $a,b \in [n]$, $\partial_{x_{a}}\partial_{x_{b}}(f_{m}) \neq 0$.
\item Conclude that for all $m \in [k]$, $f_{m}$ must be a
  multiplicative ROP. That is, the ROF computing $f_{m}$ does not
  contain any addition gate. 
\item Use the multiplicative property of $f_{k}$ to show that $f_{k}$
  can be eliminated by taking partial derivative with respect to one
  variable and substituting another by a field constant. If this
  constant is non-zero, we contradict the inductive hypothesis.
\item Otherwise, use the sum of (multiplicative) ROPs representation of
  $\M_n^{\alpha,\beta}$ to show that the degree of $f$ can be made at
  most $(n-2)$ by setting one of the variables to zero. This
  contradicts our choice of $f$ since $\beta \neq 0$.
\end{enumerate}
We now proceed with the proof.
\begin{claimn} \label{cl:lower-bound}
For all $m \in [k]$ and $a,b \in [n]$,
$\partial_{x_{a}}\partial_{x_{b}}(f_{m}) \neq 0$. 
\end{claimn}
\begin{proof}
  Suppose to the contrary that
  $\partial_{x_{a}}\partial_{x_{b}}(f_{m}) =0$. Assume without loss of generality that
$a=n$, $b=n-1$, $m=k$, so 
  $\partial_{x_{n}}\partial_{x_{n-1}}(f_{k}) = 0$.
Then, 
\begin{align*}
\begin{aligned}
\M_{n}^{\alpha,\beta} &= f = \displaystyle\sum_{m=0}^{k} f_{m} &&\text{(by assumption)} \\
\partial_{x_{n}}\partial_{x_{n-1}}(\M_{n}^{\alpha,\beta}) &=
\displaystyle\sum_{m=0}^{k} \partial_{x_{n}}\partial_{x_{n-1}} (f_{m})
&&\text{(by additivity of partial derivative)} \\ 
\M_{n-2}^{\alpha,\beta} &= \displaystyle\sum_{m=0}^{k-1}
\partial_{x_{n}}\partial_{x_{n-1}} (f_{m}) &&\text{(by recursive
structure of $\M_n$,}\\
&&& \text{and since
  $\partial_{x_{n}}\partial_{x_{n-1}}(f_{k}) = 0$)} 
\end{aligned}
\end{align*}
Thus $\M_{n-2}^{\alpha,\beta}$ can be written as the sum of $k-1$
polynomials, each of which is a ROP
(by Fact~\ref{fact:ROP}).
By the inductive hypothesis, $2(k-1) \geq
(n-2)$. Therefore, $k \geq n/2$ contradicting our assumption. 
\qed \end{proof}

From Claim~\ref{cl:lower-bound} and Fact~\ref{fact:mrop-partial}, we
can conclude: 

\begin{observation}
  \label{obs:lower-bound}
For all $m \in [k]$, $f_{m}$ is a multiplicative ROP.
\end{observation}
Observation \ref{obs:lower-bound} and Lemma \ref{lem:mrops}
together imply that
for each $m \in [k]$ and $a\in [n]$, there exist $b \neq a
\in [n]$ and $\gamma \in \F$ such that
$\partial_{x_{b}}(f_{m})\mid_{x_{a} = \gamma} = 0$.  There are two
cases to consider.

First, consider the case when for some $m,a$ and the corresponding
$b,\gamma$, it turns out that $\gamma\neq 0$.  Assume without loss of
generality that $m=k$, $a=n-1$, $b=n$, so that
$\partial_{x_{n}}(f_{k})\mid_{x_{n-1} = \gamma} = 0$. (For other
indices the argument is symmetric.) Then
\begin{align*}
\begin{aligned}
\M_{n}^{\alpha,\beta} &= \displaystyle\sum_{i \in [k]} f_{i} &&\text{(by
  assumption)} \\ 
\partial_{x_{n}}(\M_{n}^{\alpha,\beta})\mid_{x_{n-1} = \gamma} &=
\displaystyle\sum_{i \in [k]} \partial_{x_{n}}(f_{i})\mid_{x_{n-1} =
  \gamma} &&\text{(by additivity of partial derivative)} \\ 
\M_{n-1}^{\alpha,\beta}\mid_{x_{n-1} = \gamma} &= \displaystyle\sum_{i \in
  [k-1]} \partial_{x_{n}}(f_{i})\mid_{x_{n-1} = \gamma} &&\text{(since
  $\gamma$ is chosen as per
  Lemma \ref{lem:mrops})} \\ 
\M_{n-2}^{\alpha\gamma+\beta,\beta\gamma} &= \displaystyle\sum_{i \in [k-1]}
\partial_{x_{n}}(f_{i})\mid_{x_{n-1} = \gamma} && \text{(recursive
  structure of $\M_n$)}
\end{aligned}
\end{align*}
Therefore, $\M_{n-2}^{\alpha\gamma+\beta,\beta\gamma}$ can be
written as a sum of at most $k-1$ polynomials, each of which is a ROP
(Fact~\ref{fact:ROP}). By the inductive hypothesis, 
$2(k-1) \geq n-2$ implying that $k \geq n/2$ contradicting our
assumption.

(Note: the term $\M_{n-2}^{\alpha\gamma+\beta,\beta\gamma}$ is what
necessitates a stronger induction hypothesis than working with just
$\alpha=0, \beta=1$.) 

It remains to handle the case when for all $m\in [k]$ and $a\in [n]$,
the corresponding value of $\gamma$ to some $x_b$ (as guaranteed by
Lemma~\ref{lem:mrops}) is $0$. Examining
the proof of Lemma~\ref{lem:mrops}, this implies
that each leaf node in any of the ROFs can be made zero only by setting
the corresponding variable to zero. That is, the linear forms at all
leaves are of the form $a_ix_i$. 

Since each $\phi_{m}$ is a multiplicative ROP, setting $x_n = 0$ makes
the variables in the polynomial computed at the sibling of the leaf
node $a_nx_n$ redundant. Hence setting $x_n=0$ reduces the degree of
each $f_m$ by at least 2. That is, $\Deg(f\mid_{x_n=0}) \leq n-2$. But
$\M_{n}^{\alpha,\beta}\mid_{x_n=0}$ equals $\M_{n-1}^{\beta,0} = \beta S_{n-1}^{n-1}$,
which has degree $n-1$, contradicting the asusmption that 
$f = \M_{n}^{\alpha,\beta}$. 
\qed\end{proof}

The following lemma shows that the above lower bound is indeed optimal.
\begin{lemma} \label{lem:upper-bound}
For any field $\F$ and $\alpha, \beta \in \F$, the polynomial $f = \alpha
S^{n}_{n} + \beta S^{n-1}_{n}$ can be written as a  sum of at most $\lceil
n/2 \rceil$ ROPs.  
\end{lemma}

\begin{proof}
For $n$ odd, this follows immediately from
Proposition~\ref{prop:easy-upper-bound2}. 

If $n$ is even, say $n = 2k$, then define the following polynomials:
\begin{eqnarray*}
\textrm{for~} i\in [k-1], ~~f_{i} &=& (x_{2i-1} + x_{2i}) \cdot
\left(\displaystyle\prod_{\substack{ k \in [n] \\ k \neq 2i, 2i-1 }}
x_{k} \right) \\
f_k &=& \left( \beta x_{2k-1} + \beta x_{2k} + \alpha x_{2k-1} x_{2k} \right) \cdot
\left(\displaystyle\prod_{\substack{ m \in [n] \\ k \neq 2k, 2k-1 }}
x_{m} \right).
\end{eqnarray*}
Then we have $ f = \beta (f_{1} + f_{2} + \ldots + f_{k-1}) + f_k$.

Note that each $f_i$ is an ROP; for $i<k$ this is immediate, and for
$i=k$, the factor involving $x_{2k-1}$ and
$x_{2k}$ is bivariate multilinear and hence an ROP.
Thus we have a representation 
of $f$ as a sum of  $k = \lceil n/2 \rceil$ ROPs. 
\qed \end{proof}

Combining the results of Lemma~\ref{lem:lower-bound} and
Lemma~\ref{lem:upper-bound}, we obtain the following theorem. 
At $\alpha=0, \beta=1$, it yields Theorem~\ref{thm:hierarchy}.
\begin{theorem}
\label{thm:hierarchy-general}
For each $n \ge 1$, any
$\alpha \in \F$ and any any $\beta \in \F \setminus \{0\}$, the
polynomial $\alpha S_{n}^{n} + \beta S_{n}^{n-1}$ is in $\sum^k \cdot
\textrm{ROP}$ but not in $\sum^{k-1} \cdot \textrm{ROP}$, where $k=
\lceil n/2 \rceil$. 
\end{theorem}

\section{A 4-variate multilinear polynomial not in $\sum^2\cdot \textrm{ROP}$}
\label{sec:hard-4-variate}
This section is devoted to proving Theorem~\ref{thm:hard-for-2-sum}.
We want to find an explicit 4-variate multilinear polynomial that is
not expressible as the sum of 2 ROPs.

Note that the proof of Theorem~\ref{thm:hierarchy} does not help here,
since the polynomials separating $\sum^2\cdot \textrm{ROP}$ from
$\sum^3\cdot \textrm{ROP}$ have 5 or 6 variables. One obvious approach
is to consider other combinations of the symmetric polynomials. This
fails too; we can show that all such combinations are in $\sum^2 \cdot
\ROP$. 
\begin{proposition}
\label{prop:4-variate-sympoly-combo}
For every choice of field constants $a_i$ for each
$i \in \{0,1,2,3,4\}$, the polynomial $\sum_{i=0}^4 a_i S_4^i$ can be
expressed as the sum of two ROPs.
\end{proposition}
\begin{proof}
Let $g = \sum_i a_{i} S_4^i$. 
We obtain the expression for $g$ in different ways in 4 different
cases. 
\[
\begin{array}{|l@{\hspace{5mm}}|@{\hspace{5mm}}rl|}
\hline 
\textrm{Case} & \multicolumn{2}{c|}{\textrm{Expression} }
\\ \hline 
a_2=a_3 = 0 & 
g = & a_0 + a_1 S_4^1 + a_4 S_4^4 
\\ \hline 
a_2=0; & 
g = & \left(a_1 + a_3 x_1 x_2) (x_3 + x_4 + \frac{a_4}{a_3}x_3x_4)\right)
\\
a_3 \neq 0 && + \left((a_1 + a_3 x_3x_4) (x_1 + x_2 - \frac{a_1a_4}{a_3^2})\right) +
c \\ \hline 
a_2\neq0; & 
a_2g = & (a_1 + a_2(x_1 + x_2) + a_3 x_1 x_2) (a_1 + a_2(x_3 + x_4) +
a_3 x_3 x_4)  \\
a_2a_4 = a_3^2 & & + \left(a_2^2 - a_1a_3) (x_1x_2 + x_3x_4)\right) + c \\ \hline 
a_2\neq 0; & 
a_2 g = & (a_1 + a_2(x_1 + x_2) + a_3 x_1 x_2) (a_1 + a_2(x_3 +
x_4) + a_3 x_3 x_4)  \\ 
a_2a_4 \neq a_3^2  & & + \left(x_1x_2 + \frac{a_2^2 -
a_1a_3}{a_2a_4 - a_3^2}\right) \left((a_2a_4 - a_3^2) x_3x_4 + a_2^2 - a_1a_3\right) +
c  \\ \hline 
\end{array}
\]
In the above, $c$ is an appropriate field constant, and can be added
to any ROP. Notice that the first expression is a sum of two ROPs
since it is the sum of a linear polynomial and a single monomial. All
the other expressions have two summands, each of which is a product of
variable-disjoint bivariate polynomials (ignoring constant
terms). Since every bivariate polynomial is a ROP, these
representations are also sums of $2$ ROPs.
\qed\end{proof}

Instead, we define a polynomial that gives carefully chosen weights to
the monomials of $S_4^2$. Let $f^{\alpha,\beta,\gamma}$ denote the
following polynomial:
\[f^{\alpha,\beta,\gamma} = \alpha \cdot (x_{1}x_{2} + x_{3}x_{4}) + \beta \cdot
(x_{1}x_{3} + x_{2}x_{4}) + \gamma \cdot (x_{1}x_{4} +
x_{2}x_{3}).\] To keep notation simple, we will omit the superscript
when it is clear from the context. In the theorem below, we
obtain necessary and sufficient conditions on $\alpha,\beta,\gamma$
under which $f$ can be expressed as a sum of two ROPs. 

\begin{theorem}[Hardness of representation for sum of
  $2$ ROPs] \label{thm:hard-4-variate} Let $f$ be the polynomial
  $f^{\alpha,\beta,\gamma} = \alpha \cdot (x_{1}x_{2} + x_{3}x_{4})
  + \beta
\cdot (x_{1}x_{3} + x_{2}x_{4}) + \gamma \cdot (x_{1}x_{4} +
x_{2}x_{3})$. The following are equivalent:
\begin{enumerate}
\item $f$ is not expressible as the sum of two ROPs.
\item $\alpha, \beta, \gamma$ satisfy all the three conditions C1, C2,
  C3 listed below.
  \begin{description}
    \item[C1:] $\alpha\beta\gamma \neq 0$. 
    \item[C2:] $(\alpha^2-\beta^2)(\beta^2 -
      \gamma^2)(\gamma^2-\alpha^2) \neq 0$. 
    \item[C3:] None of the equations $X^2 - d_i = 0$, $i\in [3]$,
  has a root in $\F$, where
      \begin{eqnarray*}
        d_1 &=& (+\alpha^2 - \beta^2 - \gamma^2)^2 - (2\beta\gamma)^2 \\
        d_2 &=& (-\alpha^2 + \beta^2 - \gamma^2)^2 - (2\alpha\gamma)^2 \\
        d_3 &=& (-\alpha^2 - \beta^2 + \gamma^2)^2 - (2\alpha\beta)^2 \\
      \end{eqnarray*}
  \end{description}
\end{enumerate}
 \end{theorem}

\begin{remark}
  \begin{enumerate} \item It follows, for instance, that $2(x_{1}x_{2}
    + x_{3}x_{4}) + 4(x_{1}x_{3} + x_{2}x_{4}) + 5(x_{1}x_{4} +
    x_{2}x_{3})$ cannot be written as a sum of $2$ ROPs over reals,
    yielding Theorem~\ref{thm:hard-for-2-sum}.

   \item If $\F$ is an algebraically closed field, then for every
      $\alpha,\beta, \gamma$, condition C3 fails, and so every
      $f^{\alpha,\beta,\gamma}$ can be written as a sum of 2 ROPs.
      However we do not know if there are other examples, or whether
      all multilinear 4-variate polynomials are expressible as the sum
      of two ROPs.

  \item Even if $\F$ is not algebraically closed, condition C3 fails
      if for each $a \in \F$, the equation $X^2=a$ has a root.
  \end{enumerate}
\end{remark}

Our strategy for proving Theorem~\ref{thm:hard-4-variate} is a generalization
of an idea used in \cite{Vol16}. While Volkovich showed that 3-variate ROPs
have a nice structural property in terms of their partial derivatives and commutators,
we show that the sums of two 4-variate ROPs have at least one nice structural
property in terms of their bivariate restrictions, partial
derivatives, and commutators. Then we show that  provided 
$\alpha,\beta,\gamma$ are chosen carefully, the polynomial
$f^{\alpha,\beta,\gamma}$ will not satisfy any of these properties and
hence cannot be a sum of two ROPs.

To prove Theorem~\ref{thm:hard-4-variate}, we first consider the
easier direction, $1 \Rightarrow 2$, and prove the contrapositive.
\begin{lemma}
\label{lem:4-var-upper-bound}
If $\alpha, \beta, \gamma$ do not satisfy all of C1,C2,C3, then the
polynomial $f$ can be written as a sum of 2 ROPs.
\end{lemma}
\begin{proof}
\noindent \textbf{C1 false:} If any of $\alpha, \beta, \gamma$ is
zero, then by definition $f$ is the the sum of at most two ROPs.

\noindent \textbf{C2 false:} Without loss of generality,
assume $\alpha^2 = \beta^2$, so $\alpha= \pm \beta$. Then $f$ is computed
by $f = \alpha \cdot (x_{1} \pm x_{4})(x_2 \pm x_3) + \gamma \cdot
(x_{1}x_{4} + x_{2}x_{3})$.

\noindent \textbf{C1 true; C3 false:} Without loss of generality,
the equation $X^2 - d_1 = 0$ has a root $\tau$. 
We try to express $f$ as
\[
\alpha (x_1 - a x_3) ( x_2 - b x_4) + 
\beta (x_1 - c x_2) (x_3 - d x_4).
\] 
The coefficients for $x_3x_4$ and $x_2x_4$ force
$ab=1$, $cd=1$, giving the form
\[
\alpha (x_1 - a x_3) ( x_2 - \frac{1}{a} x_4) + 
\beta (x_1 - c x_2) (x_3 - \frac{1}{c} x_4).
\] 
Comparing the coefficients for $x_1x_4$ and $x_2x_3$, we 
obtain the constraints 
\[
-\frac{\alpha}{a} - \frac{\beta}{c} = \gamma; ~~~~~~ - \alpha a
 - \beta c = \gamma
 \]
Expressing $a$ as $\frac{-\gamma - \beta c}{\alpha}$, we get a
quadratic constraint on $c$; it must be a  root of the equation
\[Z^2 + \frac{-\alpha^2+\beta^2+\gamma^2}{\beta\gamma}Z + 1 = 0. \]
Using the fact that $\tau^2 = d_1 = (-\alpha^2+\beta^2+\gamma^2)^2 -
(2\beta\gamma)^2$, we see that indeed this equation does have
roots. The left-hand size splits into linear factors, giving
\[(Z-\delta)(Z-\frac{1}{\delta}) = 0 \textrm{~~where~~}
\delta = \frac{\alpha^2-\beta^2-\gamma^2 + \tau}{2\beta\gamma}.
\] 
It is easy to verify that $\delta \neq 0$ and $\delta \neq
-\frac{\gamma}{\beta}$ (since $\alpha \neq 0$). Further, define $\mu
= \frac{-(\gamma+\beta\delta)}{\alpha}$. Then $\mu$ is well-defined
(because $\alpha\neq 0$) and is also non-zero. Now setting $c=\delta$
and $a=\mu$, we have satisfied all the constraints and so we can write
$f$ as the sum of 2 ROPs as follows:
\[
f = \alpha (x_1 - \mu x_3) ( x_2 - \frac{1}{\mu} x_4) + 
\beta (x_1 - \delta x_2) (x_3 - \frac{1}{\delta} x_4).
\] 
\qed\end{proof}

Now we consider the harder direction: $2 \Rightarrow 1$. Again, we
consider the contrapositive.  We first show
(Lemma~\ref{lem:structure-2-sum-4-variate}) a structural property
satisfied by every polynomial in $\sum^2\cdot \ROP$: it must satisfy
at least one of the three properties $C1', C2', C3'$ described in the
lemma. We then show (Lemma~\ref{lem:no-structure}) that under the
conditions $C1,C2,C3$ from the theorem statement, $f$ does not satisfy any of
$C1', C2', C3'$; it follows that $f$ is not expressible as the sum of
2 ROPs.

\begin{lemma}
  \label{lem:structure-2-sum-4-variate}
Let $g$ be a $4$-variate multilinear polynomial over the field $\F$
which can be expressed as a sum of $2$ ROPs. Then at least one of the
following conditions is true: 
\begin{description}
\item[C1':]  There exist $i, j \in [4]$ and $a, b \in \F$ such that
  $g\mid_{x_{i}= a, x_{j}= b}$ is linear.
\item[C2':] There exist $i, j \in [4]$ such that $x_{i}, x_{j},
  \partial_{x_{i}}(g), \partial_{x_{j}}(g), 1$ are linearly dependent. 
\item[C3':] $g = l_{1} \cdot l_{2} + l_{3} \cdot l_{4}$ where $l_{i}$s are
  linear forms,  $l_1$ and $l_2$ are   variable-disjoint, and
  $l_3$ and  $l_4$ are   variable-disjoint.  
\end{description}
\end{lemma}
\begin{proof}
Let $\phi$ be a sum of $2$ ROFs computing $g$. Let $v_{1}$ and $v_{2}$
be the children of the topmost $+$ gate. The proof is in two
steps. First, we reduce to the case when $|\Var(v_{1})| = |\Var(v_{2})|
= 4$. Then we use a case analysis to show that at least one of the
aforementioned conditions hold true. 
In both steps, we will repeatedly use
Proposition~\ref{prop:3-var-ROP}, which showed that any $3$-variate
ROP can be reduced to a linear polynomial by substituting a single
variable with a field constant. We now proceed with the proof. 

Suppose $|\Var(v_{1})| \leq 3$. Applying
Proposition~\ref{prop:3-var-ROP} first to $v_{1}$ and then to the
resulting restriction of $v_{2}$, one can see that there exist $i,
j \in [4]$ and $a, b \in \F$ such that $g\mid_{x_{i}=a, x_{j}=b}$ is a
linear polynomial. So condition $C1'$ is satisfied.

Now assume that $|\Var(v_{1})| = |\Var(v_{2})| = 4$.
Depending on the type of gates of $v_{1}$ and $v_{2}$, we consider $3$
cases.
\medskip
 
\noindent
\textbf{Case 1:} Both $v_{1}$ and $v_{2}$ are $\times$ gates. Then $g$
can be represented as $ M_{1}\cdot M_{2} + M_{3} \cdot M_{4}$ where
$(M_{1}, M_{2})$ and $(M_{3}, M_{4})$ are variable-disjoint
ROPs.

Suppose that for some $i$,  $|\Var(M_{i})| = 1$. Then,
$g \mid_{M_{i} \to 0}$ is a $3$-variate restriction of $f$ and is
clearly an ROP. Applying Proposition~\ref{prop:3-var-ROP} to this
restriction, we see that condition $C1'$ holds.

Otherwise each $M_{i}$ has $|\Var(M_{i})| = 2$.

Suppose $(M_{1}, M_{2})$ and $(M_{3}, M_{4})$ define distinct
partitions of the variable set.  Assume without loss of generality that $g = M_{1}(x_{1},
x_{2}) \cdot M_{2}(x_{3}, x_{4}) + M_{3}(x_{1}, x_{3}) \cdot
M_{4}(x_{2}, x_{4})$.  If all $M_{i}$s are linear forms, it is clear
that condition $C3'$ holds. If not, assume that $M_{1}$ is of the form
$l_{1}(x_{1}) \cdot m_{1}(x_{2}) + c_{1}$ where $l_{1}, m_{1}$ are
linear forms and $c_{1} \in \F$. Now $g\mid_{l_{1} \to 0} = c_{1}
\cdot M_{2}(x_{3}, x_{4}) + M_{3}'(x_{3}) \cdot M_{4}(x_{2},
x_{4})$. Either set $x_{3}$ to make $M_{3}'$ zero, or, if that is not
possible because $M_{3}'$ is a
non-zero field constant, then set $x_{4} \to b$ where $b \in \F$.  In both
cases, by setting at most 2 variables, we obtain a linear polynomial,
so $C1'$ holds.   

Otherwise, $(M_{1}, M_{2})$ and $(M_{3}, M_{4})$ define the same
partition of the variable set. Assume without loss of generality that $g = M_{1}(x_{1},
x_{2}) \cdot M_{2}(x_{3}, x_{4}) + M_{3}(x_{1}, x_{2}) \cdot
M_{4}(x_{3}, x_{4})$. If one of the $M_{i}$s is linear, say without loss of generality that
$M_{1}$ is a linear form, then $g \mid_{M_{4} \to 0}$ is a 2-variate
restriction which is also a linear form, so $C1'$ holds. Otherwise,
none of the $M_{i}$s is a linear form. Then each $M_{i}$ can be
represented as $l_{i} \cdot m_{i} + c_{i}$ where $l_{i}, m_{i}$ are
univariate linear forms and $c_{i} \in \F$. We consider a $2$-variate
restriction which sets $l_{1}$ and $m_{4}$ to $0$. (Note that
$\Var(l_{1}) \cap \Var(m_{4}) = \emptyset$.) Then the resulting
polynomial is a linear form, so $C1'$ holds.
\medskip

\noindent
\textbf{Case 2:} Both $v_{1}$ and $v_{2}$ are $+$ gates. Then $g$ can
be written as $f = M_{1} + M_{2} + M_{3} + M_{4}$ where $(M_{1},
M_{2})$ and $(M_{3}, M_{4})$ are variable-disjoint ROPs.

Suppose $(M_{1}, M_{2})$ and $(M_{3}, M_{4})$ define distinct
partitions of the variable set.

Suppose further that there exists $M_{i}$ such that $|\Var(M_{i})| =
1$.  Without loss Of generality, $\Var(M_1) = \{x_1\}$, $\{x_1,x_2\} \subseteq \Var(M_3)$,
and $x_3 \in \Var(M_4)$. 
Any setting to $x_2$ and $x_4$ results in a
linear polynomial, so $C1'$ holds.

So assume without loss of generality that $g = M_{1}(x_{1}, x_{2}) + M_{2}(x_{3}, x_{4}) +
 M_{3}(x_{1}, x_{3}) + M_{4}(x_{2}, x_{4})$. Then for $a, b \in \F$,
 $g\mid_{x_{1}=a, x_{4}=b}$ is a linear polynomial, so $C1'$ holds.
 
Otherwise, $(M_{1}, M_{2})$ and $(M_{3}, M_{4})$ define the same
partition of the variable set. Again, if say $|\Var(M_{1})| =1$, then
setting two variables from $M_2$ shows that $C1'$ holds. So assume
without loss of generality that $g = M_{1}(x_{1}, x_{2}) + M_{2}(x_{3}, x_{4}) +
M_{3}(x_{1}, x_{2}) + M_{4}(x_{3}, x_{4})$. Then for $a, b \in \F$,
$g\mid_{x_{1}=a, x_{3}=b}$ is a linear polynomial, so again $C1'$
holds.
\medskip

\noindent
\textbf{Case 3:} One of $v_{1}, v_{2}$ is a $+$ gate and the other is
a $\times$ gate. Then $g$ can be written as $g = M_{1} + M_{2} + M_{3}
\cdot M_{4}$ where $(M_{1}, M_{2})$ and $(M_{3}, M_{4})$ are
variable-disjoint ROPs. Suppose that $|\Var(M_{3})| = 1$. Then
$g \mid_{ M_{3} \to 0}$ is a $3$-variate restriction which is a
ROP. Using Proposition~\ref{prop:3-var-ROP}, we get a $2$-variate
restriction of $g$ which is also linear, so $C1'$ holds. The same
argument works when $|\Var(M_{4})| = 1$. So assume that $M_{3}$ and
$M_{4}$ are bivariate polynomials.

Suppose that  $(M_{1}, M_{2})$ and $(M_{3}, M_{4})$ define distinct
partitions of the variable set. 
Assume without loss of generality that $g = M_{1} + M_{2} + M_{3}(x_{1}, x_{2}) \cdot
M_{4}(x_{3}, x_{4})$, and $x_3, x_4$ are separated by $M_1, M_2$.
Then $g\mid_{M_{3} \to 0}$ is a $2$-variate restriction which is also
linear, so $C1'$ holds.
  
Otherwise $(M_{1}, M_{2})$ and $(M_{3}, M_{4})$ define the same
partition of the variable set. Assume without loss of generality that  $g = M_{1}(x_{1},
x_{2}) + M_{2}(x_{3}, x_{4}) + M_{3}(x_{1}, x_{2}) \cdot M_{4}(x_{3},
x_{4})$.  If $M_{1}$ (or $M_{2}$) is a linear form, then consider a
$2$-variate restriction of $g$ which sets $M_{4}$ (or $M_{3}$) to
$0$. The resulting polynomial is a linear form. Similarly if $M_{3}$
(or $M_{4}$) is of the form $l \cdot m + c$ where $l, m$ are
univariate linear forms, then we consider a $2$-variate restriction
which sets $l$ to $0$ and some $x_{i} \in \Var(M_{4})$ to a field
constant. The resulting polynomial again is a linear form.  In all
these cases, $C1'$  holds. 
 
The only case that remains is that $M_{3}$ and $M_{4}$ are linear
forms while $M_{1}$ and $M_{2}$ are not. Assume that
$M_1 = (a_1x_1+b_1)(a_2x_2+b_2)+c$ and
$M_{3} = a_{3} x_{1} + b_{3} x_{2} + c_{3}$. Then
$\partial_{x_{1}}(g) = a_1(a_2x_{2}+b_2) + a_{3} M_{4}$  and
$\partial_{x_{2}}(g) = (a_1 x_{1} +b_1)a_2 + b_{3} M_{4}$.
It follows that
$b_{3} \cdot \partial_{x_{1}}(g) - a_{3} \cdot \partial_{x_{2}}(g) +
a_1a_2a_3 x_1 - a_1a_2b_3x_2 = a_1b_2b_3-b_1a_2a_3 \in \F$, 
and hence the polynomials $x_{1}$, $x_{2}$,
$\partial_{x_{1}}(g)$, $\partial_{x_{2}}(g)$ and $1$ are linearly
dependent. Therefore, condition $C2'$ 
of the lemma is satisfied. 
\qed \end{proof}

\begin{lemma}
  \label{lem:no-structure}
If $\alpha, \beta, \gamma$ satisfy conditions $C1,C2,C3$ from the
statement of Theorem~\ref{thm:hard-4-variate}, then the polynomial
$f^{\alpha,\beta,\gamma}$ does not satisfy any of the properties
$C1',C2',C3'$ from Lemma~\ref{lem:structure-2-sum-4-variate}.
\end{lemma}
\begin{proof}
  \noindent $\mathbf{C1 \Rightarrow \neg C1'}$: Since
  $\alpha\beta\gamma\neq 0$, $f$ contains all possible degree $2$
  monomials. Hence after setting $x_i=a$ and $x_j=b$, the monomial
  $x_{k}x_{l}$ where $k,l \in [4] \backslash \{i, j\}$ still survives.
\medskip

  \noindent $\mathbf{C2 \Rightarrow \neg C2'}$: The proof is by
  contradiction. Assume to the contrary that for some $i,j$, without loss of generality say
  for $i=1$ and $j=2$, the polynomials $x_{1}, x_{2},
  \partial_{x_{1}}(f), \partial_{x_{2}}(f), 1$ are linearly dependent.
  Note that $\partial_{x_{1}}(f) = \alpha x_{2} + \beta x_{3} + \gamma
  x_{4}$ and $\partial_{x_{2}}(f) = \alpha x_{1} + \gamma x_{3} +
  \beta x_{4}$. This implies that the vectors $ (1, 0, 0, 0, 0)$, $(0, 1,
  0, 0, 0)$, $(0, \alpha, \beta, \gamma, 0)$, $(\alpha, 0, \gamma,
  \beta, 0)$ and $(0,0,0,0,1)$ are linearly dependent. This further implies that the
  vectors $(\beta, \gamma)$ and $(\gamma, \beta)$ are linearly
  dependent. Therefore, $\beta = \pm \gamma$, contradicting C2.
  \medskip

  \noindent $\mathbf{C1 \wedge C2 \wedge C3 \Rightarrow \neg C3'}$:
  Suppose, to the contrary, that $C3'$ holds. That is, $f$ can be written as
  $f = l_{1} \cdot l_{2} + l_{3} \cdot l_{4}$ where
  $(l_{1}, l_{2})$ and $(l_{3}, l_{4})$
  are variable-disjoint linear forms. By the preceding arguments, we
  know that $f$ does not satisfy $C1'$ or $C2'$. 

First consider the case when $(l_{1}, l_{2})$ and $(l_{3}, l_{4})$
define the same partition of the variable set. Assume without loss of generality that
$\Var(l_1) = \Var(l_3)$, 
$\Var(l_2) = \Var(l_4)$, and $|\Var(l_1)| \le 2$. Setting the
variables in $l_1$ to any field constants yields a linear form, so $f$
satisfies C1', a contradiction.  

Hence it must be the case that $(l_{1}, l_{2})$ and $(l_{3}, l_{4})$
define different partitions of the variable set. Since all degree-2
monomials are present in $f$, each pair $x_i$, $x_j$ must be separated
by at least one of the two partitions. This implies that both
partitions have exactly 2 variables in each part.  Assume without loss
of generality that $f = l_{1}(x_{1}, x_{2}) \cdot l_{2}(x_{3}, x_{4})
+ l_{3}(x_{1}, x_{3}) \cdot l_{4}(x_{2}, x_{4})$.

At this point, we use properties of the commutator of $f$; recall
Definition~\ref{def:commutator}.  By
Lemma \ref{lem:commutator}, we know that  $l_{2}$ divides
$\triangle_{12}f$. We compute $\triangle_{12}f$ explicitly for our
candidate polynomial: 
\begin{align*}
\begin{aligned}
 \triangle_{12}f &= ( \alpha x_{3} x_{4})(\alpha + (\beta +
 \gamma)(x_{3} + x_{4}) + \alpha x_{3} x_{4}) \\
 &~~~
 - (\beta x_{4} + \gamma x_{3} + \alpha x_{3}x_{4})(\beta x_{3} +
 \gamma x_{4} + \alpha  x_{3}x_{4}) \\ 
 &= -\beta\gamma(x_{3}^{2} + x_{4}^{2}) +
 (\alpha^{2}-\beta^{2}-\gamma^{2}) x_{3}x_{4}\\ 
 \end{aligned}
 \end{align*}
Since $l_2$ divides $\triangle_{12}f$, $\triangle_{12}f$ is not
irreducible but is the product of two linear factors. Since
$\triangle_{12}f(0, 0) = 0$, at least one of the linear factors of
$\triangle_{12}f$ must vanish at $(0, 0)$. Let $x_{3} - \delta x_{4}$
be such a factor. Then $\triangle_{12}(f)$ vanishes not only at
$(0,0)$, but whenever $x_3 = \delta x_4$. Substituting $x_{3} = \delta
x_{4}$ in $\triangle_{12}f$, we get
 \begin{align*}
 \begin{aligned}
   -\delta^{2}\beta\gamma - \beta\gamma
   + \delta(\alpha^{2} - \beta^{2}  - \gamma^{2}) = 0 
 \end{aligned}
 \end{align*}
 Hence $\delta$ is of the form
 \[
\delta = \frac{- (\alpha^{2} - \beta^{2}  - \gamma^{2}) \pm \sqrt{(\alpha^{2}
 - \beta^{2}  - \gamma^{2})^2 - 4\beta^2\gamma^2 } }{-2\beta\gamma}
 \]
Hence $2\beta\gamma\delta -
(\alpha^{2} - \beta^{2}  - \gamma^{2})$ is a root of the equation
 $X^2-d_1 = 0$, contradicting the assumption that C3 holds. 

Hence it must  be the case that $C3'$ does not hold.
\qed\end{proof}

With this, the proof of Theorem~\ref{thm:hard-4-variate} is complete.

The conditions imposed on $\alpha, \beta, \gamma$ in
Theorem~\ref{thm:hard-4-variate} are tight and irredundant. Below we
give some explicit examples over the field of reals.

\begin{enumerate}
\item $f = 2(x_{1}x_{2} + x_{3}x_{4}) + 2(x_{1}x_{3} + x_{2}x_{4}) +
  3(x_{1}x_{4} + x_{2}x_{3})$ satisfies conditions
  C1 and C3 from the Theorem but not C2; $\alpha=\beta$. A
  $\sum^2\cdot \ROP$ representation for $f$ is $f = 2(x_{1} +
  x_{4})(x_{2} + x_{3}) + 3(x_{1}x_{4} + x_{2}x_{3})$.
\item $f = 2(x_{1}x_{2} + x_{3}x_{4}) - 2(x_{1}x_{3} + x_{2}x_{4}) +
  3(x_{1}x_{4} + x_{2}x_{3})$ satisfies conditions
  C1 and C3 but not C2; $\alpha=-\beta$. A
  $\sum^2\cdot \ROP$ representation for $f$ is $f = 2(x_{1} -
  x_{4})(x_{2} - x_{3}) +   3(x_{1}x_{4} + x_{2}x_{3})$.
\item $f = (x_{1}x_{2} + x_{3}x_{4}) + 2(x_{1}x_{3} + x_{2}x_{4}) +
  3(x_{1}x_{4} + x_{2}x_{3})$  satisfies conditions
  C1 and C2 but not C3. A
  $\sum^2\cdot \ROP$ representation for $f$ is 
  $f = (x_{1} + x_{3})(x_{2} + x_{4}) +
  2(x_{1} + x_{2}) (x_{3} + x_{4})$.
\end{enumerate}

\section{Conclusions}
\label{sec:concl}

\begin{enumerate}
\item We have seen in Proposition~\ref{prop:easy-upper-bound} that
  every $n$-variate multilinear polynomial ($n \ge 4$) can be written
  as the sum of $3 \times 2^{n-4}$ ROPs. A counting argument shows
  that there exist multilinear polynomials $f$ requiring exponentially
  many ROPs summands; if $f \in \sum^k\cdot \ROP$ then $k =
  \Omega(2^n/n^2)$.  Our general upper bound on $k$ is $O(2^n)$,
  leaving a small gap between the lower and upper bound.  What is the
  true tight bound? Can we find explicit polynomials that require
  exponentially large $k$ in any $\sum^{k}\cdot \ROP$ expression?
\item We have shown in Theorem~\ref{thm:hierarchy} that for each $k$,
  $\sum^k\cdot \ROP$ can be separated from $\sum^{k-1}\cdot \ROP$ by a
  polynomial on $2k-1$ variables.  Can we separate these classes with
  fewer variables? Note that any separating polynomial must have
  $\Omega(\log k)$ variables.
\item In particular, can 4-variate multilinear polynomials separate
  sums of 3 ROPs from sums of 2 ROPs over every field? If not, what is
  an explicit example?
\end{enumerate}

\bibliographystyle{plain}
\bibliography{HOR}

\end{document}